\documentclass[a4paper,12pt]{article}

\usepackage[ansinew]{inputenc}

\usepackage{amsfonts,amssymb,amsmath,amsthm}

\newtheorem{thm}{Theorem}[section]
\newtheorem{lem}{Lemma}[section]
\newtheorem{cor}{Corollary}[section]

\newtheorem{rem}{Remark}[section]

\numberwithin{equation}{section}

\begin{document}

\title{Scattering solutions and Born approximation for the magnetic Schr\"odinger operator} 

\author{Valery Serov$^{1}$ and Jan Sandhu$^{2}$ \\ Department of Mathematical Sciences \\ University of Oulu, Finland}

\date{}

\footnotetext[1]{E-mail address: vserov@cc.oulu.fi}
\footnotetext[2]{E-mail address: jan.sandhu@oulu.fi}

\maketitle

\begin{abstract}
We prove the existence of scattering solutions for multidimensional magnetic Schr\"odinger equation which belong to the weighted Sobolev space $H^1_{-\delta}(R^n)(n=2,3)$ with some $\delta>\frac{1}{2}$. 
As a consequence of this we formulate the direct Born approximation for the magnetic Schr\"odinger operator. Possible connections with inverse problems (inverse scattering Born approximation) are discussed.
\end{abstract}

\section{Introduction}

The main goal of present article is to justify the application of the classical direct scattering Born approximation for the magnetic Schr\"odinger operator. The direct Born approximation is known as the most applicable approximate method in the numerous practical problems. It is also known that the inverse scattering Born approximation is well-defined and perfectly works (as the mathematical tool) in the case of linear and nonlinear Schr\"odinger operators and for all types of scattering data: full scattering, backscattering, fixed angle scattering and fixed energy scattering. For some scattering data it is possible to get the uniqueness and reconstruction procedure while for some data we are able to reconstruct singularities and jumps of unknowns even when there is no uniqueness. We mention here the results of P\"aiv\"arinta and Somersalo \cite{PSo}, Nachman \cite{N1}, \cite{N2}, Sun and Uhlmann \cite{SuU}, Isakov and Sylvester \cite{ISy}, P\"aiv\"arinta, Serov and Somersalo \cite{PSeSo}, P\"aiv\"arinta and Serov \cite{PSe1}, \cite{PSe3}, \cite{PSe4}, Ola, P\"aiv\"arinta and Serov \cite{OPSe}, Ruiz \cite{Ru}, Ruiz and Vargas \cite{RuV}, P\"aiv\"arinta and Serov \cite{PSe2}, Reyes \cite{Re}, Serov \cite{Se}, Serov and Harju \cite{SeH1}, \cite{SeH2}, Serov and Sandhu \cite{SeSa}, Lechleiter \cite{Le}, Reyes and Ruiz \cite{ReRu} and some others. The main point of all these results is the precise calculation of the first (quadratic) nonlinear term in the Born series. For the magnetic Schr\"odinger operator the direct scattering problem (i.e., existence of the scattering solutions) as well as the inverse scattering Born approximation are not familiar at all.
The big interest to this problem is connected to the fact that the knowledge of the scattering amplitude with backscattering data allows us to obtain essential information about the unknowns.

We consider the magnetic Schr\"odinger operator
\begin{equation}
H=-(\nabla+i \vec{W}(x))^2+V(x)\cdot,\quad x\in R^n,
\end{equation}   
in dimensions $n=2,3,$ where the coefficients $\vec{W}(x)$ and $V(x)$ are assumed to be real-valued. We assume generally that $\vec{W}(x)\in L^{\infty}_{\delta}(R^n)$ and
\begin{equation}
\nabla\vec{W}(x)\in L^p_{\delta}(R^n),\quad n=3,\quad 3\le p\le \infty;\quad n=2,\quad 2<p\le \infty
\end{equation}
and
\begin{equation}
V(x) \in L^p_{\delta}(R^n),\quad n=3,\quad 3\le p\le \infty;\quad n=2,\quad 2<p\le \infty,
\end{equation}
where $\delta>\frac{n+1}{2}-\frac{n}{p}$. Here $L^p_{\sigma}$ denotes usual weighted Lebesgue space and Sobolev space $W^1_{p,\sigma}$ is understood so that $f$ belongs to $W^1_{p,\sigma}(R^n)$ if and only if $f$ and $\nabla f$ belong to $L^p_{\sigma}(R^n)$. For the case $p=2$ instead of the symbol $W^1_{2,\sigma}$ we use the symbol $H^1_{\sigma}$.

It is well-known that under these conditions for the coefficients of the magnetic Schr\"odinger operator the following G\aa rding's inequality holds:
$$
(H u,u)_{L^2(R^n)}\ge \nu\|\nabla u\|^2_{L^2(R^n)}-C\|u\|^2_{L^2(R^n)},
$$
where $0<\nu<1, C>0$. This inequality allows us to define symmetric operator $H$ by the method of quadratic forms. $H$ has a self-adjoint Friedrichs extension with the domain (in general)
$$
D(H)=\{f(x)\in W^1_2(R^n): H f(x)\in L^2(R^n)\}.
$$
In our particular case it is possible to prove that actually
$$
D(H)=W^2_2(R^n).
$$

In the scattering theory the main role are played by the special solutions of the equation
$$
H u(x)=k^2 u(x)
$$
which are of the form
$$
u(x)=u_0(x)+u_{sc}(x),
$$
where $u_0(x)=e^{ik(x,\theta)}$ is incident wave with direction $\theta \in S^{n-1}$ and the scattered wave $u_{sc}(x)$
satisfies the Sommerfeld radiation condition at the infinity, i.e.
\begin{equation}
\lim\limits_{r\to +\infty}r^{\frac{n-1}{2}}\left(\frac{\partial u_{sc}(x)}{\partial r}-iku_{sc}(x)\right)=0,\quad r=|x|.
\end{equation}
In this case the total field $u$ satisfies the so-called Lippmann-Schwinger equation
\begin{equation}
u=u_0+\int\limits_{R^n}G_k^{+}(|x-y|)\left(i\nabla (\vec{W}(y)u)+i\vec{W}(y)\nabla u-\tilde{q}(y)u\right)\,dy,
\end{equation}
where $\tilde{q}=|\vec{W}|^2+V$ and $G_k^+$ is the kernel of integral operator $(-\Delta-k^2-i0)^{-1}$. Using the representation $u=u_0+u_{sc}$ we rewrite this integral equation (1.5) only for scattered field $u_{sc}$ as
\begin{equation}
u_{sc}=\tilde{u}_0+\int\limits_{R^n}G_k^{+}(|x-y|)\left(i\nabla (\vec{W}(y)u_{sc})+i\vec{W}(y)\nabla u_{sc}-\tilde{q}(y)u_{sc}\right)\,dy,
\end{equation} 
where $\tilde{u}_0$ is equal to 
\begin{equation}
\tilde{u}_0(x)=\int\limits_{R^n}G_k^{+}(|x-y|)\left(i\nabla (\vec{W}(y)u_0)+i\vec{W}(y)\nabla u_0-\tilde{q}(y)u_0\right)\,dy.
\end{equation}
We use the following results of Agmon \cite{Ag} (see Remark 2, Appendix A):
$$
\frac{1}{|k|}\|f\|_{H^2_{-\delta}(R^n)}+\|f\|_{H^1_{-\delta}(R^n)}+|k|\|f\|_{L^2_{-\delta}(R^n)}\le c\|(\Delta+k^2)f\|_{L^2_{\delta}(R^n)},\quad |k|\ge 1,
$$
where $\delta >\frac{1}{2}$ and $H^2_{-\delta}(R^n)$ denotes the weighted Sobolev space. As a consequence we have (for fixed $k$) that 
$$
\|(-\Delta-k^2-i0)^{-1}f\|_{H^2_{-\delta}(R^n)}\le c(k)\|f\|_{L^2_{\delta}(R^n)},
$$
and uniformly in $|k|\ge 1$ we have that
\begin{equation}
\|(-\Delta-k^2-i0)^{-1}f\|_{L^2_{-\delta}(R^n)}\le \frac{c}{|k|}\|f\|_{L^2_{\delta}(R^n)}.
\end{equation}
But since $(-\Delta-k^2-i0)^{-1}$ is the integral operator of convolution type we can conclude that for fixed $k$ it maps continuously $H^{-1}_{\delta}(R^n)$ to $H^1_{-\delta}(R^n)$,
where $H^{-1}_{\delta}(R^n)$ denotes the dual of the Sobolev space $H^1_{-\delta}(R^n)$. 

We rewrite (1.6) as the integral equation 
$$
u_{sc}=\tilde{u}_0+L_k(u_{sc}), \quad \tilde{u}_0=L_k(u_0),
$$
where the integral operator $L_k$ is defined as
\begin{equation}
L_kf(x):=\int\limits_{R^n}G_k^{+}(|x-y|)\left(i\nabla (\vec{W}(y)f)+i\vec{W}(y)\nabla f-\tilde{q}(y)f\right)\,dy.
\end{equation}
The main result of present article is Theorem 2.1 which provides the existence of the scattering solutions for the magnetic Schr\"odinger operator. This theorem is proved in Section 2. Based on the main result we justify in Section 3 the direct Born approximation for such operators.  

\section{Existence of the scattering solutions}
We are preceding a proof of the main result by the following lemmas.
\begin{lem}
Suppose that conditions (1.2) and (1.3) are fulfilled. Then there is $\delta_0>\frac{1}{2}$ such that $\tilde{u}_0\in H^1_{-\delta_0}(R^n)$ and the integral operator $L_k$ maps $H^1_{-\delta_0}(R^n)$ into itself. 
\end{lem}
\begin{proof}
Conditions for $p$ and $\delta$ from (1.2) and (1.3) imply that there is $\delta_0>\frac{1}{2}$ such that
$$
L^p_{\delta}(R^n)\subset L^2_{\delta_0}(R^n).
$$ 
It is therefore true that under the conditions (1.2) and (1.3) functions $V$, $\vec{W}$, $\nabla \vec{W}$ and $|\vec{W}|^2$ belong to $L^2_{\delta_0}(R^n)$ with the same $\delta_0$. Since $u_0$ is a bounded and smooth function we may conclude (using Agmon's result (1.8)) that $\tilde{u}_0$ belongs to $H^1_{-\delta_0}(R^n)$. It can be mentioned here that we have no longer uniform estimates in $k$ as in (1.8). In order to prove that $L_k$ maps $H^1_{-\delta_0}(R^n)$ into itself we note that if $f$ belongs to $H^1_{-\delta_0}(R^n)$ then $(1+|x|^2)^{-\frac{\delta_0}{2}}f$ belongs to usual Sobolev space $H^1(R^n)$. Using now Sobolev imbedding theorem we conclude that
$$
f\in L^{\frac{2n}{n-2}}_{-\delta_0}(R^n), \quad n=3,\quad f\in L^s_{-\delta_0}(R^n), \quad s<\infty,\quad n=2. 
$$
Then the conditions (1.2) and (1.3) and H\"older inequality allow us easily conclude that $\tilde{q}f$ and $(\nabla\vec{W})f$ belong $L^2_{\delta_0}(R^n)$. Since we have $\vec{W}(x)\in L^{\infty}_{\delta}(R^n)$ the function $\vec{W}\nabla f$ will belong to $L^2_{\delta_0}(R^n)$ too. The final step is the application of Agmon's result (1.8). 
\end{proof}
We may prove a little bit more about this operator $L_k$.
\begin{lem} Let us assume that $\vec{W}\in L^{\infty}_{\delta}(R^n)$, 
\begin{equation}
\nabla \vec{W}(x)\in L^p_{\delta}(R^n),\quad n=3,\quad 3\le p\le \infty;\quad n=2,\quad 2<p\le \infty,
\end{equation}
where $\delta>\frac{n+1}{2}-\frac{n}{p}$, and
\begin{equation}
V(x) \in L^p_{loc}(R^n),\quad n=3,\quad 3\le p\le \infty;\quad n=2,\quad 2<p\le \infty,
\end{equation}
and that $\vec{W}$ and $V$ have special behavior at the infinity such that
\begin{equation}
|V(x)|,\quad |\vec{W}(x)|,\quad |\nabla\vec{W}(x)|\le \frac{c}{|x|^{\mu}},\quad |x|\to \infty,
\end{equation} 
where $\mu>2$ for $n=2,3$. Then the operator $L_k$ is compact in $H^1_{-\delta_0}(R^n)$ for some $\delta_0>\frac{1}{2}$.
\end{lem}
\begin{proof} Let us choose $R>0$ large enough and represent $V, \vec{W}$ and $\nabla\vec{W}$ as 
$$
V=V_1+V_2,\quad \vec{W}=\vec{W}_1+\vec{W}_2,\quad \nabla\vec{W}=\nabla\vec{W}_1+\nabla\vec{W}_2,
$$  
where the supports of the functions $V_1$, $\vec{W}_1$ and $\nabla\vec{W}_1$ are included in the ball $B_R=\{x\in R^n: |x|\le R\}$, but supports of the functions $V_2$, $\vec{W}_2$ and $\nabla\vec{W}_2$ are included in the set $\{x\in R^n: |x|\ge R\}$. Without loss of generality we assume in addition that the functions $V_2$, $\vec{W}_2$ and $\nabla\vec{W}_2$ are continuous and satisfy the conditions (2.3) for all $|x|\ge R$.

Conditions (2.1)-(2.2) and the previous considerations imply that
$$
i\nabla (\vec{W}_1(x)f(x))+i\vec{W}_1(x)\nabla f(x)-\tilde{q}_1(x)f(x) \in L^2(B_R)
$$
for any function $f\in H^1_{-\delta_0}(R^n)$. But $L^2(B_R)$ is compactly imbedded in $H^{-1}(B_R)$ and therefore in $H^{-1}_{\delta_0}(R^n)$ for $\delta_0>\frac{1}{2}$. It remains to mention now that due to Agmon's result (1.8) operator $(-\Delta-k^2-i0)^{-1}$ maps continuously the space $H^{-1}_{\delta_0}(R^n)$ to the space $H^{1}_{-\delta_0}(R^n)$. Thus, part of the operator $L_k$ that corresponds to $V_1,\vec{W}_1$ and $\nabla\vec{W}_1$ is compact in $H^{1}_{-\delta_0}(R^n)$ for $\delta_0>\frac{1}{2}$. Since outside the ball $B_R$ these functions $V_2$, $\vec{W}_2$ and $\nabla\vec{W}_2$ satisfy the conditions (2.3) we may firstly conclude using Sobolev imbedding theorem that for $f\in H^1_{-\delta_0}(R^n)$ we have
$$
i\nabla (\vec{W}_2(x)f(x))+i\vec{W}_2(x)\nabla f(x)-\tilde{q}_2(x)f(x) \in L^2_{\delta_0}(R^n)
$$
for some $\delta_0$ if and only if $\mu>2\delta_0+1$. But under the conditions of Lemma 2.2 this criterion is satisfied if $\delta_0>\frac{1}{2}$ is chosen appropriately.

Since the conditions (2.3) are satisfied we can find two sequences $\phi_j(x)\in C^{\infty}_0(R^n\setminus B_R)$ and $\psi_j(x)$ (vector-valued) $\in C^{\infty}_0(R^n\setminus B_R)$ such that
$$
\|\phi_j-V_2\|_{L_{\delta'}^{\infty}(R^n\setminus B_R)}\to 0,\quad \|\psi_j-\vec{W}_2\|_{L_{\delta'}^{\infty}(R^n\setminus B_R)}\to 0,
$$ 
$$
\|\nabla\psi_j-\nabla\vec{W}_2\|_{L_{\delta'}^{\infty}(R^n\setminus B_R)}\to 0
$$
as $j\to \infty$ for any $\delta'<\mu$.
These approximation properties imply that
$$
\|i\nabla ((\vec{W}_2-\psi_j)f)+i(\vec{W}_2-\psi_j)\nabla f-(\tilde{q}_2-Q_j)f\|_{L^2_{\delta_0}(R^n)}\le 
$$
\begin{equation}
\le c\||\nabla (\vec{W}_2-\psi_j)|+|\vec{W}_2-\psi_j|+|\tilde{q}_2-Q_j|\|_{L^{\infty}_{2\delta_0}(R^n)}\|f\|_{H^1_{-\delta_0}(R^n)}\to 0
\end{equation}
as $j\to \infty$, where $Q_j=|\psi_j|^2+\phi_j$. Since we can choose $\delta_0>\frac{1}{2}$ and $\mu$ from condition (2.3) such that $2\delta_0<\mu$ then (2.4) means that part of the operator $L_k$ which corresponds to $V_2,\vec{W}_2$ and $\nabla\vec{W}_2$ is compact in $H^{1}_{-\delta_0}(R^n)$ too. 
\end{proof}
\begin{lem}
Under the same assumptions as in Lemma 2.2 for any fixed $k>0$ and for any $f\in H^1_{-\delta_0}(R^n)$ with some $\delta_0>\frac{1}{2}$ the following asymptotical representation holds:
$$
L_kf(x)=c_n\frac{e^{ik|x|}k^{\frac{n-3}{2}}}{|x|^{\frac{n-1}{2}}}\int\limits_{R^n}e^{-ik(\theta',y)}\left(i\nabla (\vec{W}(y)f)+i\vec{W}(y)\nabla f-\tilde{q}(y)f\right)\,dy+
$$
\begin{equation}
+o\left(\frac{1}{|x|^{\frac{n-1}{2}}}\right),\quad |x|\to \infty,
\end{equation}
where $\theta'=\frac{x}{|x|}$.
\end{lem}
\begin{proof}
In view of (1.9) one must study the behavior for $|x|\to \infty$ of the function
$$
G_k^+(|x-y|)=\frac{i}{4}\left(\frac{k}{2\pi|x-y|}\right)^{\frac{n-2}{2}}H^{(1)}_{\frac{n-2}{2}}(k|x-y|),
$$
where $H^{(1)}_{\frac{n-2}{2}}$ denotes the Hankel function of first kind and of order $\frac{n-2}{2}$. In order to do that we take two cases, $k|x-y|>1$ and $k|x-y|<1$. For the first case we use the behavior of the Hankel function $H^{(1)}_{\frac{n-2}{2}}$ for large argument (see, for example, \cite{L}), i.e. 
$$
H^{(1)}_{\frac{n-2}{2}}(z)=c_n\frac{e^{iz}}{\sqrt{z}}+O \left(\frac{1}{z^{\frac{3}{2}}}\right),\quad z\to +\infty.
$$
So that ($k>0$ is fixed) we have that,
$$
G_k^+(k|x-y|)=c_n\frac{e^{ik|x-y|}k^{\frac{n-3}{2}}}{|x-y|^{\frac{n-1}{2}}}+O\left(\frac{1}{|x-y|^{\frac{n+1}{2}}}\right), \quad |x|\to +\infty.
$$
Consider two subcases, $|y|\le |x|^a$ and $|y|\ge |x|^a$, where $0<a<\frac{1}{2}$ is a parameter. In the first case we have (since $a<\frac{1}{2}$)
$$
|x-y|^{-\frac{n-1}{2}}=|x|^{-\frac{n-1}{2}}(1+O(|x|^{a-1})),\quad |x|\to +\infty.
$$
That is why we have for $|y|\le |x|^a$ with $0<a<\frac{1}{2}$ that
$$
|x-y|^{-\frac{n-1}{2}}e^{ik|x-y|}=\frac{e^{ik|x|}e^{-ik(\theta',y)}}{|x|^{\frac{n-1}{2}}}+O (|x|^{\frac{n+1}{2}-2a}), \quad \theta'=\frac{x}{|x|},
$$
as $|x|\to +\infty$. Substituting this asymptotic to the integral
$$
S_1:=\int\limits_{k|x-y|>1}G_k^+(|x-y|)\left(i\nabla (\vec{W}(y)f)+i\vec{W}(y)\nabla f-\tilde{q}(y)f\right)\,dy
$$
gives that
$$
S_1=\int\limits_{k|x-y|>1,|y|\le |x|^a}G_k^+(|x-y|)\left(i\nabla (\vec{W}(y)f)+i\vec{W}(y)\nabla f-\tilde{q}(y)f\right)\,dy+
$$
$$
+\int\limits_{k|x-y|>1,|y|\ge |x|^a}G_k^+(|x-y|)\left(i\nabla (\vec{W}(y)f)+i\vec{W}(y)\nabla f-\tilde{q}(y)f\right)\,dy=
$$
$$
=c_n\frac{e^{ik|x|}k^{\frac{n-3}{2}}}{|x|^{\frac{n-1}{2}}}\int\limits_{k|x-y|>1,|y|\le |x|^a}e^{-ik(\theta',y)}\left(i\nabla (\vec{W}(y)f)+i\vec{W}(y)\nabla f-\tilde{q}(y)f\right)\,dy+
$$
$$
+\int\limits_{k|x-y|>1,|y|\ge |x|^a}G_k^+(|x-y|)\left(i\nabla (\vec{W}(y)f)+i\vec{W}(y)\nabla f-\tilde{q}(y)f\right)\,dy+
$$
\begin{equation}
+\int\limits_{k|x-y|>1,|y|\le |x|^a}O (|x|^{\frac{n+1}{2}-2a})\left(i\nabla (\vec{W}(y)f)+i\vec{W}(y)\nabla f-\tilde{q}(y)f\right)\,dy.
\end{equation}
The conditions (2.1)-(2.3) allow us easily conclude that for any $f\in H^1_{-\delta_0}(R^n)$ the integrand $i\nabla (\vec{W}(y)f)+i\vec{W}(y)\nabla f-\tilde{q}(y)f$ belongs to $L^1(R^n)$. Hence, the last term in the latter sum is $o(|x|^{-\frac{n-1}{2}})$ since $0<a<\frac{1}{2}$. Denoting the second term in the latter sum (2.6) by $I$ we may estimate it (using conditions (2.3)) as follows:
$$
|I|\le C\int\limits_{k|x-y|>1,|y|\ge |x|^a}\frac{|i\nabla (\vec{W}(y)f)+i\vec{W}(y)\nabla f-\tilde{q}(y)f|}{|x-y|^{\frac{n-1}{2}}}\,dy\le 
$$
$$
\le \frac{C}{|x|^{\frac{n-1}{2}}}\int\limits_{|x|^a \le |y|\le \frac{|x|}{2}}|i\nabla (\vec{W}(y)f)+i\vec{W}(y)\nabla f-\tilde{q}(y)f|\,dy+
$$
$$
+C\int\limits_{|y|\ge \frac{|x|}{2}}\frac{|\nabla \vec{W}(y)||f|+|\vec{W}(y)||\nabla f|+|\tilde{q}(y)||f|}{|x-y|^{\frac{n-1}{2}}}\,dy\le 
$$
$$
\le o\left(\frac{1}{|x|^{\frac{n-1}{2}}}\right)+C\int\limits_{|y|\ge \frac{|x|}{2}}\frac{|f|+|\nabla f|}{|x-y|^{\frac{n-1}{2}}|y|^{\mu}}\,dy.
$$
Since $f\in H^1_{-\delta_0}(R^n)$ the latter inequality implies
\begin{equation}
|I|\le o\left(\frac{1}{|x|^{\frac{n-1}{2}}}\right)+\frac{C}{|x|^{\frac{n-1}{2}}}\left(\int\limits_{|y|\ge \frac{|x|}{2}}\frac{1}{|x-y|^{n-1}|y|^{2\mu-(n-1)-2\delta_0}}\,dy\right)^{\frac{1}{2}}\|f\|_{H^1_{-\delta_0}(R^n)}.
\end{equation}
Since $\mu>\frac{n+1}{2}$ then $\delta_0>\frac{1}{2}$ can be chosen here such that the last integral in (2.7) might be considered as the convolution of "weak singularities" and therefore we have
\begin{equation}
I=o\left(\frac{1}{|x|^{\frac{n-1}{2}}}\right).
\end{equation}
The first case $k|x-y|>1$ is thus completely investigated. 

In order to consider the second case $k|x-y|<1$ we use the behavior of the Hankel function $H^{(1)}_{\frac{n-2}{2}}$ for small argument (see \cite{L})
$$
H^{(1)}_{\frac{n-2}{2}}(z)=\left\{\begin{array}{rcl}cz^{-\frac{n-2}{2}},\quad n>2,\\
                                                    c(1+\log(z)),\quad n=2,\\
                                                    \end{array}\right. \quad z\to 0.
$$
Let us consider first $n=3$. Then using this asymptotic and taking into account that for fixed $k>0$ and large $|x|$ it can be assumed that $|y|\ge \frac{|x|}{2}$, one can estimate the integral
$$
S_2:=\int\limits_{k|x-y|<1}G_k^+(|x-y|)\left(i\nabla (\vec{W}(y)f)+i\vec{W}(y)\nabla f-\tilde{q}(y)f\right)\,dy 
$$
as
$$
|S_2|\le C\left(\int\limits_{k|x-y|<1}|G_k^+(k|x-y|)|^2\,dy\right)^{\frac{1}{2}}\left(\int\limits_{|y|\ge \frac{|x|}{2}}\frac{|\nabla f|^2+|f|^2}{|y|^{2\mu}}\,dy\right)^{\frac{1}{2}} \le 
$$
\begin{equation}
\le \frac{C}{|x|^{\mu-\delta_0}}\left(\int\limits_{k|x-y|<1}|x-y|^{-2}\,dy\right)^{\frac{1}{2}}\|f\|_{H^1_{-\delta_0}(R^n)}=o\left(\frac{1}{|x|}\right),
\end{equation}
for $|x|\to +\infty$, if $\delta_0$ is chosen such that $\frac{1}{2}<\delta_0<1$. In the same way we have that
$$
S_2=o\left(\frac{1}{|x|^{\frac{1}{2}}}\right),\quad |x|\to +\infty
$$
in the two dimensional case due to the behavior of the Hankel function $H_0^{(1)}$ for small argument. Combining (2.6), (2.8) and (2.9) we get (2.5). 
\end{proof}
\begin{rem}
The proof of the last lemma shows that the function $L_kf(x)$ is continuous for all $x$ such that $|x|\ge R,$ where $R$ is large enough.
\end{rem}
These lemmas allow us to obtain the main result of this work.
\begin{thm}
Under the same assumptions as in Lemma 2.2 and for any $k\ne 0$ the integral equation (1.6) has a unique scattering solution from the space $H^1_{-\delta_0}(R^n)$ for some $\delta_0>\frac{1}{2}$.
\end{thm}
\begin{proof}
Since the operator $L_k$ is compact in the space $H^1_{-\delta_0}(R^n)$ we can apply the Riesz theory in this Hilbert space. Based on this methodology we will prove that the
homogeneous equation $f-L_k f=0$ has only the trivial solution in the space $H^1_{-\delta_0}(R^n)$ (i.e. the operator $I-L_k$ is injective). But this will imply that the operator is also surjective and the inverse $(I-L_k)^{-1}$ is bounded in $H^1_{-\delta_0}(R^n)$.
This condition is equivalent to the claim that the equation (1.6) has a unique solution from the space $H^1_{-\delta_0}(R^n)$. 

It is possible to check that any $f\in H^1_{-\delta_0}(R^n)$ which satisfies the homogeneous equation $f-L_k f=0$ belongs to $H^2_{loc}(R^n)$ and satisfies also the equation 
$$
H f=k^2 f
$$
and Sommerfeld radiation condition (1.4). These facts imply that 
$$
0=\int\limits_{|x|\le R}\nabla\left(\overline{f}(\nabla+i \vec{W})f-f(\nabla-i \vec{W})\overline{f}\right)\,dx.
$$
The divergence theorem then gives that
$$
0=\int\limits_{|x|=R}(\overline{f}\partial_{\nu}f-f\partial_{\nu}\overline{f})\,d\sigma(x) + 2i\int\limits_{|x|=R}|f|^2\vec{W}\nu\,d\sigma(x),
$$
where $\nu$ denotes the normal vector at the boundary of the ball $B_R$. The radiation conditions (1.4) for the functions $f$ and $\overline{f}$ allow us to conclude that
$$
0=2i\int\limits_{|x|=R}|f|^2\vec{W}\nu\,d\sigma(x)+2ik\int\limits_{|x|=R}|f|^2\,d\sigma(x)+o\left(\frac{1}{R^{\frac{n-1}{2}}}\right)\int\limits_{|x|=R}f\,d\sigma(x).
$$ 
Now the assumption that the function $f\in H^1_{-\delta_0}(R^n)$ satisfies the homogeneous equation $f-L_k f=0$ implies by the equation (2.5) that 
$$
f=O\left(\frac{1}{|x|^{\frac{n-1}{2}}}\right),\quad |x|\to \infty.
$$ 
This behavior and the condition (2.3) imply that
$$
\int\limits_{|x|=R}|f|^2\,d\sigma(x)=o(1), \quad R\to \infty.
$$
This fact and Lemma 2.3 (see (2.5)) allow us easily conclude that (it is enough to integrate in (2.5) with respect to $x$)
$$
\int\limits_{R^n}e^{-ik(\theta',y)}\left(i\nabla (\vec{W}(y)f(y))+i\vec{W}(y)\nabla f(y)-\tilde{q}(y)f(y)\right)\,dy=0.
$$
Thus, we have actually (since the function $f$ satisfies the homogeneous equation $f-L_k f=0$) that
\begin{equation}
f=o\left(\frac{1}{|x|^{\frac{n-1}{2}}}\right),\quad |x|\to \infty.
\end{equation}
The Sommerfeld radiation condition (1.4) implies also that 
\begin{equation}
\nabla f=o\left(\frac{1}{|x|^{\frac{n-1}{2}}}\right),\quad |x|\to \infty.
\end{equation}
Using these two facts we are going to prove that $f=0$ a.e. In order to do so we first prove that $f(x)\equiv 0$ for $|x|\ge R_0$ with $R_0$ large enough.
Indeed, it suffices to prove that, for $r\ge R_0$, the radial function
\begin{equation}
F(r):=\int\limits_{S^{n-1}}f(r\theta)\phi(\theta)\,d\theta
\end{equation}
is identically equal to zero for each eigenfunction $\phi$ of the Laplace operator $\Delta_S$ on the unit sphere $S^{n-1}$. The eigenfunctions satisfy the equations
$$
(\Delta_S+\mu^2)\phi=0,\quad \mu^2=k(k+n-2),\quad k=0,1,2,..., 
$$
where $\mu^2$ are the eigenvalues of the Laplace operator $\Delta_S$ (see, for example, \cite{Sh}).

In view of the formula for the Laplacian $\Delta$ on $R^n$ in polar coordinates
$$
\Delta=\frac{\partial^2}{\partial r^2}+\frac{n-1}{r}\frac{\partial}{\partial r}+\frac{1}{r^2}\Delta_S,
$$
it follows that $F(r)$ satisfies the ordinary differential equation
$$
F''(r)+\frac{n-1}{r}F'(r)+(k^2-\frac{\mu^2}{r^2})F(r)=
$$
\begin{equation}
=-2i\int\limits_{S^{n-1}}\vec{W}(r\theta)\nabla f(r\theta)\phi(\theta)\,d\theta-\int\limits_{S^{n-1}}\tilde{V}(r\theta)f(r\theta)\phi(\theta)\,d\theta,
\end{equation}
where $\tilde{V}=i\nabla\vec{W}+|\vec{W}|^2+V$. Let us rewrite this linear equation in the form
$$
F''(r)+\frac{n-1}{r}F'(r)+(k^2-\frac{\mu^2}{r^2})F(r)=\Phi(r,F).
$$
Since $V$ and $\vec{W}$ satisfy (2.3) then it is not difficult to check that function $\Phi(r)$ from the right-hand side of equation (2.13) for $r\to \infty$ behaves as 
\begin{equation}
\Phi(r,F)=O\left(\frac{1}{r^{\alpha}}\right),\quad \alpha>n.
\end{equation}
It is well-known (see, for example, \cite{L}) that the homogeneous equation corresponding to (2.13) has two linearly independent solutions $r^{-\frac{n-2}{2}}H^{(j)}_{\nu}(kr)$,
$j=1,2,$ where $H^{(1)}_{\nu}(z)$ and $H^{(2)}_{\nu}(z)$ are the Hankel functions of order $\nu,\nu^2=\mu^2+(\frac{n-2}{2})^2,$ and of the first and the second kind, respectively. 
In view of the asymptotic behavior of the Hankel functions (see \cite{L}) it follows that the behavior of these two solutions is of the form
\begin{equation}
r^{-\frac{n-2}{2}}H^{(j)}_{\nu}(kr)=\frac{C_j}{r^{\frac{n-1}{2}}}e^{\pm ikr}+o\left(\frac{1}{r^{\frac{n-1}{2}}}\right),\quad r\to \infty,
\end{equation}
where $+$ corresponds to $j=1$ and $-$ corresponds to $j=2$. Next, we use Green's function (see, for example, \cite{St}) for the Bessel equation with Dirichlet boundary conditions on the interval 
$[r_0,2r_0]$, where $r_0$ is large enough,
\begin{equation}
g(r,\xi)=\frac{C}{\xi^{n-1}} \left\{
                              \begin{array}{rcl}
                                  u_1(r)u_2(\xi),\quad r<\xi,\\
                                  u_1(\xi)u_2(r),\quad r>\xi.\\
                              \end{array}
                             \right.
\end{equation}                                                                                                                                
Here $u_1$ and $u_2$ are two linearly independent solutions of homogeneous equation (2.13) which satisfy
homogeneous Dirichlet boundary conditions at $r_0$ and $2r_0$, respectively, that is,
$$
u_1(r)=\left(\frac{1}{r_0r}\right)^{\frac{n-2}{2}}\left(H^{(1)}_{\nu}(kr_0)H^{(2)}_{\nu}(kr)-H^{(1)}_{\nu}(kr)H^{(2)}_{\nu}(kr_0)\right),
$$
$$
u_2(r)=\left(\frac{1}{2r_0r}\right)^{\frac{n-2}{2}}\left(H^{(1)}_{\nu}(2kr_0)H^{(2)}_{\nu}(kr)-H^{(1)}_{\nu}(kr)H^{(2)}_{\nu}(2kr_0)\right),
$$
and $C=-r_0^{n-1}u_2(r_0)u_1'(r_0)$. 
Using this Green's function, equation (2.13) can be reduced to the following
linear integral equation
$$
F(r)=K_1r^{-\frac{n-2}{2}}H^{(1)}_{\nu}(kr)+K_2r^{-\frac{n-2}{2}}H^{(2)}_{\nu}(kr)+\int\limits_{r_0}^{2r_0}g(r,\xi)\Phi(\xi,F)\,d\xi,
$$
where $K_1$ and $K_2$ are constants. This integral equation can be uniquely solved by iterations. Since (2.14) and (2.15) hold then using representation (2.16) we obtain
for the integral part of the latter integral equation the following estimate:
\begin{equation}
\int\limits_{r_0}^{2r_0}g(r,\xi)\Phi(\xi,F)\,d\xi=O\left(\frac{1}{r^{\alpha-1}}+\frac{1}{r^{n-1}r_0^{\alpha-n}}\right),\quad \alpha>n,
\end{equation}
as $r_0\to +\infty$ and $r_0\le r\le 2r_0$. Thus, the estimates (2.17) allow us to conclude that any solution $F(r)$ of the non-homogeneous equation (2.13) has asymptotic (when $r\to \infty$) that is the linear combination of two asymptotic (2.15). 
Since the hypothesis implies that $F(r)=o\left(\frac{1}{r^{\frac{n-1}{2}}}\right)$, we deduce that $V(r)\equiv 0$ for all $r$ large enough. 
Thus, the same is true for $f(x)$ (see definition (2.12)) for all $|x|\ge R_0$ with $R_0$ large enough. 

We are in the position now to apply the unique continuation principle (UCP) in $R^n$. Due to UCP for a second order elliptic differential operators with real coefficients (see Theorem 17.2.8 in \cite{H})
we may immediately conclude that $f\equiv 0$ in $R^n$. Thus, $I-L_k$ is injective and therefore, the integral equation (1.6) has a unique solution from the space $H^1_{-\delta_0}(R^n)$ which is given by the formula
\begin{equation}
u_{sc}=(I-L_k)^{-1}\tilde{u}_0\Longleftrightarrow u=u_0+(I-L_k)^{-1}L_ku_0,
\end{equation}
where $u$ is as in (1.5). Theorem 2.1 is completely proved.
\end{proof}
\begin{cor}
If the conditions (2.1)-(2.3) are satisfied then the magnetic Schr\"odinger operator $H$ has no positive eigenvalues.
\end{cor}
\begin{proof}
If $\lambda$ is a positive eigenvalue of $H$ then 
$$
Hu=\lambda u,\quad u\in L^2(R^n),\quad Hu\in L^2(R^n).
$$
It means that this $u$ belongs to the domain of the Friedrichs self-adjoint extension of $H$ and therefore $u\in W^1_2(R^n)$. This fact allows us to conclude that this $u$ satisfies the homogeneous equation
$$
u=L_{\sqrt{\lambda}}u
$$
with an integral operator from (1.9).
Thus, since $W^1_2(R^n)\subset H^1_{-\delta}(R^n)$ for $\delta>\frac{1}{2}$ (actually this imbedding holds for any $\delta>0$) we may apply to this $u$ the same proof as in Theorem 2.1 and conclude that actually $u\equiv 0$.
\end{proof}
Using Agmon's results (1.8) the operator $L_k$ can be extended to the space $L^2_{-\delta}(R^n)$ as a uniformly bounded operator with respect to $k\ge 1$ such that 
\begin{equation}
\|u_{sc}\|_{L^2_{-\delta}(R^n)}\le C, \quad k\ge 1,
\end{equation}
where constant $C$ depends only on the corresponding norms of $V$, $\vec{W}$ and $\nabla\vec{W}$. Based on this fact one can show that if the corresponding norms of $V$, $\vec{W}$ and $\nabla\vec{W}$ are small enough then the operator norm of $L_k$ as an operator from $L^2_{-\delta}(R^n)$ to itself is strictly less than $1$. In that case the formula (2.18) can be rewritten as
\begin{equation}
u=u_0+\sum\limits_{j=1}^{\infty}L^j_k(u_0).
\end{equation}  
Thus, the scattering solution $u$ can be obtained as the series of iterations of $u_0$ in the equation (1.5). 

\section{Scattering amplitude and direct backscattering Born approximation}

In this section we will consider the direct backscattering Born approximation for the magnetic Schr\"odinger operator $H$ with conditions (2.1)-(2.3). The motivation to this problem is connected to the fact that the knowledge of the scattering amplitude with the backscattering data gives essential information about the unknown function $V$ and $\vec{W}$.

Theorem 2.1 and Lemma 2.3 (see (2.5)) yield the following asymptotical representation for the scattering solutions $u(x,k,\theta)$ with fixed $k>0$ as $|x|\to +\infty$:
$$
u(x,k,\theta)=e^{ik(x,\theta)}+c_n\frac{e^{ik|x|}k^{\frac{n-3}{2}}}{|x|^{\frac{n-1}{2}}}A(k,\theta',\theta)+o\left(\frac{1}{|x|^{\frac{n-1}{2}}}\right),
$$ 
where function $A$ is called the scattering amplitude and defined by
\begin{equation}
A(k,\theta',\theta)=\int\limits_{R^n}e^{-ik(\theta',y)}\left(i\nabla (\vec{W}(y)u)+i\vec{W}(y)\nabla u-\tilde{q}(y)u\right)\,dy.
\end{equation}
Substituting $u=u_0+u_{sc}$ into the equation (3.1) gives that
$$
A(k,\theta',\theta)=\int\limits_{R^n}e^{-ik(\theta',y)}\left(i\nabla (\vec{W}(y)u_0)+i\vec{W}(y)\nabla u_0-\tilde{q}(y)u_0\right)\,dy+
$$
$$
+\int\limits_{R^n}e^{-ik(\theta',y)}\left(i\nabla (\vec{W}(y)u_{sc})+i\vec{W}(y)\nabla u_{sc}-\tilde{q}(y)u_{sc}\right)\,dy:= 
$$
\begin{equation}
:=A_B(k,\theta',\theta)+R(k,\theta',\theta).
\end{equation}
The function $A_B$ is called the direct Born approximation. It can be checked (using integration by parts) that $A_B$ is actually equal to 
$$
A_B(k,\theta',\theta)=-k(\theta+\theta')F(\vec{W})(k(\theta-\theta'))-F(\tilde{q})(k(\theta-\theta')),
$$
where $F$ denotes usual $n-$dimensional Fourier transform as
$$
F(f)(\xi)=\int\limits_{R^n}f(x)e^{i(x,\xi)}\,dx.
$$
The particular case $\theta'=-\theta$ yields the direct backscattering Born approximation 
\begin{equation}
A_B^b(k,-\theta,\theta)=-F(\tilde{q})(2k\theta).
\end{equation}


Formulae (3.2) and (3.3) show that in the frame of the Born approximation
$$
A(k,-\theta,\theta)\approx -F(|\vec{W}|^2+V)(2k\theta).
$$
But we want to write more terms in the Born series. For this purpose we calculate term $R$ in the scattering amplitude. Using the series (2.20) and integration by parts we have that
$$
R(k,-\theta,\theta)=-i\int\limits_{R^n}e^{ik(\theta,y)}\nabla\vec{W}(y)L_ku_0(y)\,dy+2k\theta\int\limits_{R^n}e^{ik(\theta,y)}\vec{W}(y)L_ku_0(y)\,dy-
$$
$$
-\int\limits_{R^n}e^{ik(\theta,y)}\tilde{q}(y)L_ku_0(y)\,dy+R_2(k,-\theta,\theta),
$$ 
where the term $R_2$ corresponds to the series $\sum\limits_{j=2}^{\infty}L^j_k(u_0)$ and equals to
$$
R_2(k,-\theta,\theta)=i\int\limits_{R^n}e^{ik(\theta,y)}\nabla\vec{W}(y)\sum\limits_{j=2}^{\infty}L^j_ku_0(y)\,dy+
$$
\begin{equation}
+2i\int\limits_{R^n}e^{ik(\theta,y)}\vec{W}(y)\nabla\left(\sum\limits_{j=2}^{\infty}L^j_ku_0(y)\right)\,dy-\int\limits_{R^n}e^{ik(\theta,y)}\tilde{q}(y)\sum\limits_{j=2}^{\infty}L^j_ku_0(y)\,dy.
\end{equation}
It will be shown that the term which correspond to $R_2$ in the definition (3.1) might be neglected because of the smallness of the operator norm $L_k$ in the space $L^2_{-\delta}(R^n)$. 

Since $L_k u_0(y)$ is equal to 
$$
\int\limits_{R^n}e^{ik(\theta,z)}G_k^{+}(|y-z|)\left(i\nabla \vec{W}(z)-2k\theta \vec{W}(z)-\tilde{q}(z)\right)\,dz,
$$
then we obtain (after some simple calculations) the following representation:
$$
R(k,-\theta,\theta):=R_1(k,-\theta,\theta)+R_2(k,-\theta,\theta)=
$$
$$
=\int\limits_{R^n}\int\limits_{R^n}e^{ik(\theta,y+z)}G_k^{+}(|y-z|)\nabla\vec{W}(y)\nabla\vec{W}(z)\,dy\,dz+
$$
$$
+4ik\int\limits_{R^n}\int\limits_{R^n}e^{ik(\theta,y+z)}G_k^{+}(|y-z|)\nabla\vec{W}(y)\theta\vec{W}(z)\,dy\,dz-
$$
$$
-4k^2\int\limits_{R^n}\int\limits_{R^n}e^{ik(\theta,y+z)}G_k^{+}(|y-z|)\theta\vec{W}(y)\theta\vec{W}(z)\,dy\,dz+
$$
$$
+\int\limits_{R^n}\int\limits_{R^n}e^{ik(\theta,y+z)}G_k^{+}(|y-z|)\tilde{q}(y)\tilde{q}(z)\,dy\,dz+R_2:=
$$
\begin{equation}
:=I_1+I_2+I_3+I_4+R_2.
\end{equation}
It can be mentioned here that this equality must be understood in the sense of tempered distributions. 

Using the facts $F(G_k^+)(\eta)=\frac{1}{\eta^2-k^2-i0}$ and $F(\phi\cdot \psi)=(2\pi)^{-n}F(\phi)\ast F(\psi)$ we can calculate the terms $I_j,j=1,2,3,4,$ more precisely as
$$
I_1=(2\pi)^{-n}\int\limits_{R^n}\frac{F(\nabla\vec{W})(k\theta+\eta)F(\nabla\vec{W})(k\theta-\eta)}{\eta^2-k^2-i0}\,d\eta,
$$
$$
I_2=4ik(2\pi)^{-n}\int\limits_{R^n}\frac{F(\nabla\vec{W})(k\theta+\eta)\theta F(\vec{W})(k\theta-\eta)}{\eta^2-k^2-i0}\,d\eta,
$$
$$
I_3=-4k^2(2\pi)^{-n}\int\limits_{R^n}\frac{\theta F(\vec{W})(k\theta+\eta)\theta F(\vec{W})(k\theta-\eta)}{\eta^2-k^2-i0}\,d\eta,
$$
\begin{equation}
I_4=(2\pi)^{-n}\int\limits_{R^n}\frac{F(\tilde{q})(k\theta+\eta)F(\tilde{q})(k\theta-\eta)}{\eta^2-k^2-i0}\,d\eta.
\end{equation}
Our next step is to neglect the term $R_2$ in (3.5) and justify this neglect. Indeed, using (1.7) and Agmon's results (1.8) the $L^2_{-\delta}$-norm of $L_ku_0$ can be estimated as (uniformly in $|k|\ge 1$)
$$
\|L_ku_0\|_{L^2_{-\delta}(R^n)}\le c\left(\|\nabla\vec{W}\|_{L^2_{\delta}(R^n)}+\|\vec{W}\|_{L^2_{\delta}(R^n)}+\|\tilde{q}\|_{L^2_{\delta}(R^n)}\right).
$$
The conditions (2.1)-(2.3) show that the right hand-side of the latter inequality is finite. Thus, there is a constant $c_0$ depending only on the $L^2_{\delta}$-norms of functions $\vec{W}$, $\nabla\vec{W}$ and $\tilde{q}$ such that uniformly in $|k|\ge 1$
\begin{equation}
\|L_ku_0\|_{L^2_{-\delta}(R^n)}\le c_0.
\end{equation}
At the same time for any function $f\in H^1_{-\delta}(R^n)$ with some $\delta>\frac{1}{2}$ we can easily obtain
\begin{equation}
\|L_kf\|_{L^2_{-\delta}(R^n)}\le c\left(\|\nabla\vec{W}\|_{L^p_{2\delta}(R^n)}+\|\vec{W}\|_{L^{\infty}_{2\delta}(R^n)}+\|\tilde{q}\|_{L^p_{2\delta}(R^n)}\right)\|f\|_{H^1_{-\delta}(R^n)},
\end{equation}
where $p$ is the same as in the conditions (2.1)-(2.2). We can rewrite (3.8) in the form of operator norm
\begin{equation}
\|L_k\|_{H^1_{-\delta}(R^n)\to L^2_{-\delta}(R^n)}\le c_1,
\end{equation}
where constant $c_1$ depends only on the norms of functions $\vec{W}$, $\nabla\vec{W}$ and $\tilde{q}$ from (3.8). Hence, we may assume that these norms are chosen so small that $c_1<1$. Now we extend the operator $L_k$ as an operator from $L^2_{-\delta}(R^n)$ to $L^2_{-\delta}(R^n)$ with the same norm estimate as in (3.9). This fact together with estimate (3.7) imply that
\begin{equation}
\|\sum\limits_{j=2}^{\infty}L^j_ku_0\|_{L^2_{-\delta}(R^n)}\le \frac{c_0c_1}{1-c_1}.
\end{equation}
Hence, the left hand-side of (3.10) can be made as small as we want if $c_1$ (and, in addition, $c_0$) are chosen small enough. This fact and duality arguments show that 
\begin{equation}
\|\nabla(\sum\limits_{j=2}^{\infty}L^j_ku_0)\|_{H^{-1}_{-\delta}(R^n)}\le \frac{c_0c_1}{1-c_1},
\end{equation}
The estimates (3.10) and (3.11) imply that one can have the term $R_2(k,-\theta,\theta)$ as small as desired uniformly in $|k|\ge 1$ and $\theta\in S^{n-1}$ if the corresponding norms of functions $\vec{W}$, $\nabla\vec{W}$ and $\tilde{q}$ (or the constants $c_0$ and $c_1$) are chosen small enough. Thus, the term $R_2(k,-\theta,\theta)$ can be neglected in the Born approximation. 

Summarizing our considerations (see (3.5)-(3.6) and (3.10)-(3.11)) we may now obtain the following direct backscattering Born approximation (more precise than (3.3)) for the magnetic Schr\"odinger operator
\begin{equation}
A(k,-\theta,\theta)\approx -F(|\vec{W}|^2+V)(2k\theta)+I_1+I_2+I_3+I_4.
\end{equation}
This formula gives us very good approximation for the backscattering amplitude $A$. It is very important that for this approximation we need to have only the magnetic potential $\vec{W}$ and electric potential $V$, but we do not need (as we can see the formula (3.1)) to have the scattering solutions $u(x,k,\theta)$ of the equation
$$
Hu(x)=k^2 u(x).
$$
This direct approximation (3.12) will be effectively used for the inverse backscattering Born approximation. Namely, due to formulas (3.6) we will be able to calculate precisely the quadratic term in the Born series that corresponds to the inverse backscattering approximation and to estimate its smoothness. This smoothness result together with (3.12) will give us the solution of the inverse backscattering problem with respect to the reconstruction of the singularities and jumps of the unknowns. These problems will be investigated carefully in the futures publications. 

\bigskip

\section*{Acknowledgments} 
This work was supported by the Academy of Finland (application number 250215, Finnish Programme
for Centres of Excellence in Research 2012-2017).

\end{document}